\documentclass[runningheads,a4paper]{llncs}

\usepackage{amssymb}
\usepackage{graphicx}
\usepackage{makeidx}  
\usepackage{amsmath}
\usepackage[T2A]{fontenc}
\usepackage{amsfonts}

\usepackage{url}
\urldef{\mailsa}\path|niknikols0@gmail.com|
\newcommand{\keywords}[1]{\par\addvspace\baselineskip
\noindent\keywordname\enspace\ignorespaces#1}

\let\emptyset\varnothing

\begin{document}

\newcommand{\Qa}{\hbox to 12 pt {$^Q$\hss $A$}}

\newtheorem{thm}{Theorem}
\newtheorem{lm}{Lemma}
\newtheorem{cor}{Corollary}
\newtheorem{prop*}{Proposition}
\newtheorem{rem}[thm]{Remark}
\newtheorem{df}{Definition}
\mainmatter  

\title{Note on level $r$ consensus}
\titlerunning{Note on level $r$ consensus}  
%
\author{Nikolay~L.~Poliakov\inst{} }
\authorrunning{N.~L.~Poliakov} 
%
%
\institute{Financial University, Moscow, Russian Federation,\\
\mailsa}

\maketitle              

\begin{abstract} We show that the hierarchy of level $r$ consensus partially collapses. In particular, any
profile $\pi\in \mathcal{P}$ that exhibits consensus of level
$(K-1)!$ around $\succ_0$ in fact exhibits consensus of level $1$
around $\succ_0$.

\end{abstract}

The concept of level $r$ consensus was introduced  in \cite{LevelR}
in the context of the metric approach in social choice theory. We
will mainly use the notation and definitions of \cite{LevelR}. Let
$A=\{1,2,\ldots, K\}$ be a set of $K > 2$ alternatives and let $N =
\{1, 2, \ldots, n\}$ be a set of individuals. Each linear order
(i.e. complete, transitive and antisymmetric binary relation) on the
set $A$ is called a \emph{preference relation}. The set of all
preference relations is denoted by $\mathcal{P}$. The
\emph{inversion metric} is the function $d:\mathcal{P}\times
\mathcal{P}\to \mathbb{R}$ defined by
$$
d(\succ,
\succ')=\frac{|(\succ\setminus\succ')\cup(\succ'\setminus\succ)|}{2}
$$
(since all preference relations in $\mathcal{P}$ have the same
cardinality we have also: \mbox{$d(\succ, \succ')=$}
$|\succ\setminus\succ'|=|\succ'\setminus\succ|$).
\par Let $\succ_0$ be a preference relation in $\mathcal{P}$.
A metric on $\mathcal{P}$ allows to determine which one of any two
preference relations is closer to a third one.  This comparison can
be extended to equal-sized sets of preferences.

\begin{df} Let $C$ and $C'$ be two disjoint nonempty subsets of $\mathcal{P}$ with
the same cardinality, and let  $\succ_0\in \mathcal{P}$ be a
preference relation on $A$. We say that $C$ is at least as close to
$\succ_0$ as $C'$, denoted by $C\geq_{\succ_0}C'$, if there is a
one-to-one function $\phi:C\to C'$ such that for all $\succ\in C$,
$d(\succ, \succ_0)\leq d(\phi(\succ), \succ_0)$. We also say that
$C$ is closer than $C'$ to $\succ_0$, denoted by $C>_{\succ_0}C'$,
if there is a one to one function $\phi:C\to C'$ such that for all
$\succ\in C$, $d(\succ, \succ_0)\leq d(\phi(\succ), \succ_0)$, with
strict inequality for at least one $\succ\in C$.
\end{df}

\par Using the concept of closeness the authors define the correspondence between
preference profiles $\pi\in \mathcal{P}^n$ and preference relations
$\succ\in \mathcal{P}$ depending on a natural parameter $r$ called
``\emph{preference profile $\pi$ exhibits consensus of level $r$
around $\succ$}''.
\par For any $\pi=(\succ_1, \succ_2, \ldots, \succ_n)\in
\mathcal{P}^n$, $\succ\in \mathcal{P}$, and $C\subseteq \mathcal{P}$
\begin{align*}
\mu_\pi(\succ)=|\{i\in \mathbb{N}:
\succ_i=\succ\}|,\,\,\,\mu_\pi(C)=|\{i\in \mathbb{N}: \succ_i\in
C\}|
\end{align*}
(obviously, $ \mu_\pi(C)=\sum_{\succ\in C}\mu_\pi(\succ)$).
\begin{df} Let $r\in \{1, 2, \ldots, \frac{K!}{2}\}$, and let $\succ_0\in \mathcal{P}$. A preference
profile $\pi\in \mathcal{P}^n$ exhibits consensus of level $r$
around $\succ_0$ if
\begin{enumerate}
\item for all disjoint subsets $C, C'$ of $\mathcal{P}$ with cardinality $r$,
$C\geq_{\succ_0}C'\rightarrow \mu_\pi(C)\geq \mu_\pi(C')$
\item there are disjoint subsets $C, C'$ of $\mathcal{P}$ with cardinality
$r$, such that $C>_{\succ_0}C'$ and $\mu_\pi(C)> \mu_\pi(C')$.
\end{enumerate}
\end{df}
\par Proposition 1 of \cite{LevelR} states that the set of profiles
that exhibit consensus of level $r+1$ around $\succ_0$ extends the
set of profiles that exhibit consensus of level $r$ around
$\succ_0$. Thus, each preference relation $\succ_0$ determines the
hierarchy of preference profiles.
\par Let a preference profile $\pi$ exhibit consensus of level $r$ around
$\succ_0$. We call $\succ_0$ a \emph{level $r$ consensus relation of
$\pi$} and simply \emph{consensus relation of $\pi$} if
$r=\frac{K!}{2}$ (the level $\frac{K!}{2}$ is the maximum level for
which this concept is nontrivial).
\par A level $r$ consensus relation $\succ_0$ of profile $\pi$ may be considered as one of
probable social binary relations on the profile $\pi$. Theorem 1 of
\cite{LevelR} states that if $n$ is odd, then each profile $\pi$
have at most one consensus relation $\succ_0$ and the consensus
relation $\succ_0$ coincides with the relation $M_{\pi}$ assigned by
the majority rule to $\pi$. This result gives an interesting
sufficient condition for  transitivity of $M_{\pi}$. Furthermore,
regardless of parity of $n$, the $\succ_0$-largest element $a_1$ is
a \emph{Condorcet winner} on $\pi$.
\par For small values of $r$, level $r$ consensus relations $\succ_0$ of profile $\pi$
have some interesting additional properties. Namely, the largest
element $a_1$ with respect $\succ_0$ is selected by any scoring
rule. A \emph{scoring rule} is characterized by a non-increasing
sequence $S=(S_1, S_2, \ldots, S_K)$ of non-negative real numbers
for which $S_1> S_K$. For $k=1,2,\ldots, K$, each individual with
the preference relation $\succ$ assigns $S_{k}$ points to the $k$-th
alternative in the linear order $\succ$. The scoring rule associated
with $S$ is the function $V_S:\mathcal{P}^n\to 2^A$ whose value at
any profile $\pi=\{\succ_1, \succ_2, \ldots, \succ_n\}$ is the set
$V_S(\pi)$ of alternatives~$a$ with the maximum total score (i.e.
with the maximum sum $\sum_{1\leq i\leq K} S_{k_i}$ where $k_i$ is
the rank of $a$ in $\succ_i$). Theorem 2 in \cite{LevelR} claims
that if a preference profile $\pi$ exhibits consensus of level
$r\leq (K-1)!$ around $\succ_0$, then the $\succ_0$-largest element
$a_1$ belongs to $V_S(\pi)$ for all scoring rules $V_S$.
\par However, the authors did not notice some combinatorial properties of the concepts introduced.
We show that the hierarchy of preference profile partially
collapses. In particular, any profile $\pi\in \mathcal{P}$ that
exhibits consensus of level $(K-1)!$ around $\succ_0$ in fact
exhibits consensus of level $1$ around $\succ_0$. Thus, it would be
desirable to slightly adjust the assumption of Theorem 2 of
\cite{LevelR}.


\begin{thm} For any natural number $K> 2$ there is a
natural number \mbox{$c\leq \frac{K(K-1)}{4}$} such that for any
natural numbers $n\geq 1$ and $r\in\{1,2,\ldots, \frac{K!}{2}-c\}$,
any preference profile $\pi\in \mathcal{P}^n$, and any linear order
$\succ_0\in \mathcal{P}$ the following conditions are equivalent
\begin{enumerate}
\item $\pi\,\text{exhibits consensus of level $r$ around $\succ_0$}$
\item $\pi\,\text{exhibits consensus of level $1$
around $\succ_0$}$.
\end{enumerate}
\end{thm}
\begin{proof} The implication $2\to 1$ follows from Proposition 1 of
\cite{LevelR}. We will prove the reverse implication. Let $\succ_0$
be a linear order in $\mathcal{P}$ and  let
$$
\mathcal{P}_k(\succ_0)=\{\succ\in \mathcal{P}: d(\succ,
\succ_0)=k\}.
$$
for any natural number $k$. Obviously, $|\mathcal{P}_k(\succ_0)|$
coincides with the number of permutations of $\{1,2,\ldots, K\}$
with $k$ inversions, i.e. with the \emph{Mahonian number} $T(K,k)$
(sequence A008302 in OEIS, see \cite{A008302}). The set
$\mathcal{P}_{\frac{K(K-1)}{2}}$ contains exactly one element. We
denote this element by $\overline{\succ}_0$:
$\mathcal{P}_{\frac{K(K-1)}{2}}=\{\overline{\succ}_0\}$.
\par Let $c'$ be the number of $k$ for which $T(K,k)$ is odd:
$$
c'=|\{k\in \mathbb{N}:  T(K,k)\equiv 1\pmod 2\}|.
$$
So, $c'\leq \frac{K(K-1)}{2}$ because $\frac{K(K-1)}{2}$ is the
maximum distance between the linear orders in $\mathcal{P}$.
Moreover, $c'$ is even because
$$
\sum_{0\leq k\leq \frac{K(K-1)}{2}}T(K,k)=K!\equiv 0\pmod 2.
$$
\par Let $c=\frac{c'}{2}$. Then the inequality \mbox{$c\leq \frac{K(K-1)}{4}$} holds.
\begin{df} For any natural number $m$ a pair $(C_1, C_2)\in 2^\mathcal{P}\times
2^\mathcal{P}$ is called \emph{$m$-balanced (around $\succ_0$)} iff
\begin{enumerate}
\item $C_1\cap C_2=\emptyset$,
\item $|C_1|=|C_2|=m$,
\item $|C_1\cap \mathcal{P}_k(\succ_0)|=|C_2\cap
\mathcal{P}_k(\succ_0)|$ for any $k=0,1,\ldots, \frac{K(K-1)}{2}$.
\end{enumerate}
\end{df}

\begin{lm} Let $\succ_1, \succ_2\in \mathcal{P}\setminus \{\succ_0, \overline{\succ}_0\}$ and $\succ_1\neq \succ_2$. Then there is a $(\frac{K!}{2}-c)$-balanced pair
$(C_1, C_2)$ for which $\succ_1\in C_1$ and $\succ_2\in C_2$.
\end{lm}
\begin{proof} Note that $T(K,k)\geq 2$ for any $k\in \{1, 2, \ldots, \frac{K(K-1)}{2}-1\}$ (this follows,
for example, from a recurrence formula for $T(K,k)$, see
\cite{A008302}). Using this fact, for each $k\in \{k\in
\mathbb{N}:T(K,k)\equiv 1\pmod 2 \}$ choose a preference relation
$\succ_{(k)}\in \mathcal{P}_k(\succ_0)\setminus\{\succ_1,
\succ_2\}$. Let
$$
\mathcal{P}'_k(\succ_0)=
\begin{cases}\mathcal{P}_k(\succ_0) \,\,\,&\text{if $T(K,k)\equiv 0$,}
\\
\mathcal{P}_k(\succ_0)\setminus\{\succ_{(k)}\} \,\,\,&\text{if
$T(K,k)\equiv 1$}
\end{cases}
\pmod 2.
$$
For each $k\in\{1, \ldots, \frac{K(K-1)}{2}-1\}$ choose a set
$C_{(k)}$ with properties
\begin{enumerate}
\item $C_{(k)}\subseteq \mathcal{P}'_k(\succ_0)$,
\item $|C_{(k)}|=\frac{|\mathcal{P}'_k(\succ_0)|}{2}$,
\item $d(\succ_1, \succ_0)=k\to \succ_1\in C_{(k)}$,
\item $\succ_2\notin C_{(k)}$.
\end{enumerate}
Let
$$
C_1= \bigcup_{1\leq k\leq\frac{K(K-1)}{2}-1}
C_{(k)}\,\,\text{and}\,\,C_2= \bigcup_{1\leq
k\leq\frac{K(K-1)}{2}-1}\mathcal{P}'_k(\succ_0)\setminus C_{(k)}.
$$
\par Obviously, items 1\--- 3 of Definition 3 hold. Lemma 2 is proved.
\end{proof}

\begin{lm}For any natural number $m$ and $m$-balanced pair $(C_1, C_2)$ there is a one-to-one function
$\phi:C_1\to C_2$ satisfying
$$
d(\succ, \succ_0)=d(\phi(\succ), \succ_0)
$$
for all $\succ \in C_1$.
\end{lm}
\begin{proof} By item 3 of Definition 3 for any $k=0,1,\ldots, \frac{K(K-1)}{2}$ there
is a one-to-one mappings $\phi_k:C_1\cap \mathcal{P}_k(\succ_0)\to
C_2\cap \mathcal{P}_k(\succ_0)$ (maybe empty if $C_1\cap
\mathcal{P}_k(\succ_0)=\emptyset$). Obviously, we can put
$\phi=\bigcup_{0\leq i\leq \frac{K(K-1)}{2}}\phi_k$. Lemma 3 is
proved.
\end{proof}
\begin{cor} For any natural number $m$ and $m$-balanced pair $(C_1, C_2)$
\begin{equation*}
C_1\geq_{\succ_0} C_2\,\,\text{and}\,\,C_2\geq_{\succ_0} C_1.
\end{equation*}
\end{cor}
\begin{proof} Let $\phi$ be a function from Lemma 2. Then
$$
d(\succ, \succ_0)=d(\phi^{-1}(\succ), \succ_0)
$$
for all $\succ \in C_2$, and it remains to recall Definition 1.
\end{proof}
\par Let $\pi\in\mathcal{P}^n$ and let $\pi$ exhibit consensus of level $r\in \{1,2,\ldots, \frac{K!}{2}-c\}$
around~$\succ_0$. By Proposition 1 of \cite{LevelR} $\pi$ exhibits
consensus of level $\frac{K!}{2}-c$ around~$\succ_0$. Our next goal
is to prove that item 1 of Definition 2 holds for the profile~$\pi$
and~$r=1$.

\begin{lm} For any different $\succ_1, \succ_2\in \mathcal{P}$
$$
d(\succ_1, \succ_0)\leq d(\succ_2, \succ_0)\to
\mu_\pi(\succ_1)\geq\mu_\pi(\succ_2).
$$
\end{lm}
\begin{proof}
\par Let $\succ_1, \succ_2\in \mathcal{P}$, $\succ_1\neq \succ_2$ and $d(\succ_1, \succ_0)\leq d(\succ_2, \succ_0)$.
\par First, let $\{\succ_1, \succ_2\}\cap \{\succ_0, \overline{\succ}_0\}=\emptyset$. Consider a $(\frac{K!}{2}-c)$-balanced pair $(C_1, C_2)$ for which
$\succ_2\in C_1$ and $\succ_1\in C_2$, and a on-to-one function
$\phi:C_1\to C_2$ satisfying
$$
d(\succ, \succ_0)=d(\phi(\succ), \succ_0)
$$
for all $\succ\in C_1$. By Definition 2 and Corollary 3 we have
\begin{equation}
\mu_{\pi}(C_1)=\mu_{\pi}(C_2).
\end{equation}

\par Let $C'_1=(C_1\setminus \{\succ_2\})\cup\{\succ_1\}$ and
$C'_2=(C_2\setminus \{\succ_1\})\cup\{\succ_2\}$. Consider the
function $\phi': C'_1\to C'_2$ defined by
$$
\phi'(\succ)=
\begin{cases}
\succ_2\,\,&\text{if $\succ = \succ_1$},
\\
\phi(\succ_2)\,\,&\text{if $\succ=\phi^{-1}(\succ_1)\neq \succ_2$},
\\
\phi(\succ)\,\,&\text{otherwise}.
\end{cases}
$$
For all $\succ\in C'_1$ we have $d(\succ, \succ_0)\leq
d(\phi'(\succ), \succ_0)$, so $C'_1\geq_{\succ_0} C'_2$ by
Definition 1. Hence, by Definition 2
\begin{equation}
\mu_\pi(C'_1)\geq \mu_\pi(C'_2).
\end{equation}

\par Since $(\forall C\subseteq \mathcal{P})\,\mu_{\pi}(C)=\sum_{\succ\in C}\mu_{\pi}(\succ)$,  we have
\begin{equation}
\mu_{\pi}(C'_1)=\mu_{\pi}(C_1)-\mu_{\pi}(\succ_2)+\mu_{\pi}(\succ_1)\,\,\text{and}\,\,\mu_{\pi}(C'_2)=\mu_{\pi}(C_2)-\mu_{\pi}(\succ_1)+\mu_{\pi}(\succ_2).
\end{equation}

Then by (1), (2) and (3)
\begin{equation*}
\mu_{\pi}(\succ_1)-\mu_{\pi}(\succ_2)\geq
\mu_{\pi}(\succ_2)-\mu_{\pi}(\succ_1),
\end{equation*}
and, finally,
$$
\mu_\pi(\succ_1)\geq \mu_\pi(\succ_2).
$$
\par For further discussion, note that this implies
\begin{equation}
d(\succ_1, \succ_0)= d(\succ_2, \succ_0)\to
\mu_\pi(\succ_1)=\mu_\pi(\succ_2).
\end{equation}
for all different $\succ_1, \succ_2\in \mathcal{P}$.
\par Consider the remaining cases.

\par Let $\succ_1=\succ_0$ and $\succ_2\neq \overline{\succ}_0$. Then
denote $C''_1=(C_1\setminus \{\succ_2\})\cup\{\succ_0\}$ and
$C''_2=(C_1\setminus \{\phi(\succ_2)\})\cup\{\succ_2\}$. Consider
the function $\phi'': C''_1\to C_2$ defined by
$$
\phi''(\succ)=
\begin{cases}
\succ_2\,\,&\text{if $\succ = \succ_0$},
\\
\phi(\succ)\,\,&\text{otherwise}.
\end{cases}
$$
For all $\succ\in C''_1$ we have $d(\succ, \succ_0)\leq
d(\phi''(\succ), \succ_0)$ and, further, $C''_1\geq_{\succ_0}
C''_2$. Reasoning as before we have
$$
\mu_\pi (\succ_0)-\mu_\pi (\succ_2)\geq \mu_\pi (\succ_2)-\mu_\pi
(\phi(\succ_2)).
$$
Since $ d(\succ_2, \succ_0)=d(\phi(\succ_2), \succ_0)$, we have
$\mu_\pi(\succ_2)=\mu_\pi(\phi(\succ_2))$ by (4). Finally,
$$\mu_\pi (\succ_0)\geq \mu_\pi (\succ_2).$$
\par In the case $\succ_2=\overline{\succ}_0$ and $\succ_1\neq\succ_0$, the arguments are similar.
\par In the latter case  $\succ_1=\succ_0$ and $\succ_2=\overline{\succ}_0$. We can choose a preference relation
$\succ^\ast\in \mathcal{P}\setminus \{\succ_0,
\overline{\succ}_0\}$. According to the above, we have
$$
\mu_\pi(\succ_1)\geq \mu_\pi(\succ^\ast)\geq \mu_\pi(\succ_2).
$$
\par Lemma 3 is proved.
\end{proof}
\par To prove the theorem it remains to show that item 2 of Definition 2 holds for the
profile~$\pi$ and $r=1$. Assume
$\mu_\pi(\overline{\succ}_0)=\emptyset$. Then, for every preference
relation $\succ$ of profile $\pi$ we have
$$
d(\succ, \succ_0)> d(\overline{\succ}_0,
\succ_0)\,\,\text{and}\,\,\mu_{\pi}(\succ)>\mu_{\pi}(\overline{\succ}_0).
$$
In the opposite case, assume that item 2 of Definition 2 is not hold
for the profile~$\pi$ and $r=1$. Then by Lemma 3 the profile $\pi$
contains the same number of all linear orders in $\mathcal{P}$.
Thus, $\pi$ does not exhibit consensus of any level, a
contradiction.
\par Theorem 1 is proved.
\end{proof}
\begin{cor} Let profile $\pi$ exhibit consensus of level $(K-1)!$ around
$\succ_0$. Then $\pi$ exhibits consensus of level $1$ around
$\succ_0$.
\end{cor}
\begin{proof} Let $K\geq 4$. Then it suffices to prove the inequality
$$
(K-1)!\leq \frac{K!}{2}-\frac{K(K-1)}{4}.
$$
This is easily by induction. For $K=3$ we can use the sufficiency of
inequality
$$
(K-1)!\leq \frac{K!}{2}-\frac{|\{k:T(K,k)=1\pmod 2\}|}{2}
$$
(for $K=3$ we have $|\{k:T(3,k)=1\pmod 2\}|=2$).
\end{proof}

\end{document}